\documentclass[11.5pt]{article}
\usepackage[margin=2cm]{geometry}
\usepackage{amsmath,amssymb,amsthm,
	    natbib,
	    graphicx,
	    tabularx,
	    booktabs,
	    color,
	    xspace,    
	  }
\usepackage{array}
\newcolumntype{L}[1]{>{\raggedright\let\newline\\\arraybackslash\hspace{0pt}}m{#1}}
\newcolumntype{C}[1]{>{\centering\let\newline\\\arraybackslash\hspace{0pt}}m{#1}}
\newcolumntype{R}[1]{>{\raggedleft\let\newline\\\arraybackslash\hspace{0pt}}m{#1}}
\usepackage{caption}
\usepackage{authblk}
\usepackage{setspace}
\usepackage{textcomp}
\usepackage[mathscr]{eucal}

\usepackage[colorlinks=true,citecolor=blue]{hyperref}
\usepackage{indentfirst}
\usepackage{mathtools,enumerate}

\usepackage{natbib}
\usepackage{tabularx}
\newtheorem{theorem}{Theorem}[section]
\newtheorem{corollary}{Corollary}

\newtheorem{proposition}{Proposition}

\theoremstyle{definition}

\theoremstyle{definition}

\usepackage{etex}
\reserveinserts{18}
\usepackage{morefloats}

\def\R{{\mathbb R}}  
 
\def\d{{\textnormal{d}}}

\begin{document}
\newcommand{\upperRomannumeral}[1]{\uppercase\expandafter{\romannumeral#1}}
\title{On  the $L_p$-quantiles for the Student $t$ distribution}
\author{M. BERNARDI}
\affil{Department of Statistical Sciences, University of Padova. Email:{mbernardi@stat.unipd.it}\\}

\author{V. BIGNOZZI}
\affil{Department of Methods and Models for Economics, Territory and Finance, Sapienza University of Rome. Email:{valeria.bignozzi@uniroma1.it}\\}
\author{L. PETRELLA}
\affil{Department of Methods and Models for Economics, Territory and Finance, Sapienza University of Rome. Email:{lea.petrella@uniroma1.it}}
\maketitle

\begin{abstract}
$L_p$-quantiles   represent an important class of generalised quantiles and are defined as the minimisers of an expected asymmetric power function, see~\cite{Chen96}. For $p=1$ and $p=2$ they correspond respectively to the quantiles and the expectiles. In his paper~\cite{Koenker93} showed that the $\tau$ quantile and the $\tau$ expectile coincide for every $\tau\in(0,1)$ for a class of rescaled Student $t$ distributions with two degrees of freedom. Here, we extend this result proving that for the  Student $t$ distribution with $p$ degrees of freedom,  the $\tau$ quantile and the $\tau$ $L_p$-quantile coincide for every $\tau\in(0,1)$ and the same holds for any affine transformation. Furthermore, we investigate the properties of  $L_p$-quantiles  and provide recursive equations for the truncated moments of the Student $t$ distribution.  
\end{abstract}

\textbf{keywords}: Expectiles; truncated moments; quantiles; risk measures; double factorial coefficient.

\section{Introduction}
 During the last years the statistical, econometric and financial literature, has seen  a growing interest in generalising the quantile tool to gather information about a random variable. Starting from the seminal paper by~\cite{KoenkerBassett78}, quantile regression has emerged as a valid alternative to the mean regression in order to deeply investigate the relation between a response variable and some covariates, especially when the tails behaviour is of concern. Both the classical and Bayesian researches have developed new tools through the years, to solve the related inferential problems. We refer the interested reader to~\cite{Koenker05} and the references therein for the classical point of view and to~\cite{YuM01},~\cite{TaddyK10},~\cite{KottasK09}, ~\cite{Bernardietal15} and~\cite{Bernardietal16} for the Bayesian one. In the same spirit, \cite{NeweyP87} and~\cite{Efron91} proposed an alternative approach, namely the expectile regression, to summarise the conditional distribution of a dependent variable given the regressors. Expectile regression, based on the asymmetric least squares estimation, is a valid alternative to quantile regression and a natural extension of the mean regression. Although quantiles have an immediate interpretation, expectile regression methods have increased during the years (see for example~\citealp{Schnabel09},~\citealp{Derossi09},~\citealp{Sobotka12},~\citealp{Guoetal13} and~\citealp{Sobotka13}). One of the reasons resides in the generalisation of the mean regression and the easiness of calculation since these methods involve a minimisation problem which is continuously differentiable. Moreover, as shown in~\cite{Jones94}, expectiles are linked to quantiles, indeed they can be interpreted as quantiles of a related distribution.
 
More general tools, which include expectiles and quantiles, have been introduced by~\cite{BrecklingChambers88} and by~\cite{Chen96} who proposed the use of $M$-quantiles and $L_p$-quantiles respectively. The $M$-quantiles extend the idea of quantiles by minimising a generic asymmetric loss function, while the  $L_p$-quantiles minimise an asymmetric power  function. 
The property of minimising a given loss function, shared by all the generalised quantiles, finds its roots in decision theory (see for instance \citealp{Savage71}) but was named elicitability by~\cite{Lambertetal08}. Recently, elicitable statistical functionals have been investigated by~\cite{Gneiting11} and~\cite{Heinrich13}. Generalised quantile regression models have also  been considered in the literature, see, for example, \cite{ChambersT06} and~\cite{Pratesi09} for the $M$-quantile framework, and~\cite{BernardiBP16} for a Bayesian approach to $L_p$-quantiles regression. 

Quantiles and their generalisations play a central role also in the financial and actuarial science literature as key tools for computing capital requirements. The quantile of a loss distribution, named Value-at-Risk (VaR), was introduced by J.P. Morgan in 1994 to measure the riskiness of a financial position. Since then it has become the most widespread risk measure for regulatory purposes, see, for example, \cite{Jorion07}. Despite its popularity the VaR, as a measure of risk, presents some important drawbacks. Firstly, it does not provide any information about losses exceeding the considered quantile; further it is not a sub-additive risk measure, meaning  that it may not incentivise portfolio diversification. For these reasons several alternatives have been considered, one of which is the Expected Shortfall (ES) which is defined as the average of the quantiles exceeding the VaR. The debate on whether VaR or ES should be used for risk management purposes is still on going and was recently questioned in two consultative documents by the~\cite{BCBS13} and the International Association of Insurance Supervisors~\cite{IAIS14}. Recently the debate on relevant properties that a risk measure should satisfy, like coherency and elicitability, (see~\citealp{Artzner99} and~\citealp{Gneiting11}, respectively) brought the attention to other risk measures based on generalised quantiles. \cite{Ziegel14} and~\cite{BelliniB15}, for example, showed that the expectile is the unique risk measure that is both coherent and elicitable, while~\cite{Taylor08} provided a methodology to estimate the ES and the VaR using expectiles. This explains the increased popularity of expectiles in the very recent quantitative risk management literature, see for instance~\cite{BelliniDB15},~\cite{JakobsonsV15} and the references therein. Since $L_p$-quantiles are natural extensions of expectiles and quantiles they received particular consideration in risk management as possible competitors of VaR; in particular, they are elicitable and have other interesting properties in terms of risk measures (see for instance~\citealp{Bellinietal14}). $L_p$-quantiles are also well known in actuarial science where they belong to the class of  zero utility premiums (see for example~\citealp{DeprezG85}), while in financial mathematics they are considered as shortfall risk measures, see for instance~\cite{FollmerS2002}. We refer the interested reader to~\cite{Embrechts14} and~\cite{Emmer16} for more detailed  summaries on the recent developments on risk measurement. 

In the present contribution we focus on the class of $L_p$-quantiles and express them in terms of moments and truncated moments. Thanks to this characterisation and to their definition as minimisers of a power loss function we are able to prove that for a Student $t$  distribution with $p$ degrees of freedom the $\tau$ quantile and the $\tau$ $L_p$-quantile coincide for every $\tau\in(0,1)$ and the same holds for any affine transformation
We expect that this property of the Student $t$ distribution may be useful for the explicit computation of its quantiles which are not available in closed form. Moreover, we investigate the properties of $L_p$-quantiles and provide recursive equations for the truncated moments of the Student $t$ distribution.

 Our results can be seen as an extension of the well known result proposed in~\cite{Koenker92} and~\cite{Koenker93} who showed that the  $\tau$ quantile and the $\tau$ expectile coincide for every $\tau\in(0,1)$ for a particular class of distributions which corresponds to rescaled Student $t$ distributions with two degrees of freedom. Some extensions of these results have been already considered recently by~\cite{Zou14} who provided a class of distributions whose $\omega(\tau)$ expectile coincide with the $\tau$ quantile for any monotone function $\omega(\cdot)$.
 
The rest of the paper is organised as follows. In Section~\ref{sec:Lpquantiles} we discuss some properties of the $L_p$-quantiles and provide their representation in terms of moments and truncated moments. Section~\ref{sec:student} gathers general results on the Student $t$ distribution and presents a recursive formula for the calculation of its truncated moments. Finally, Section~\ref{sec:Lpstudent} contains the main results of the paper on the equality of quantiles and $L_p$-quantiles for a Student $t$ distribution. The proof is rather long and it involves some ideas from combinatorial analysis therefore some technical details are postponed to the Appendix.

\section{$L_p$-quantiles}
\label{sec:Lpquantiles} 
The $L_p$-quantile of a random variable $Y$ at level $\tau\in(0,1)$, denoted by $\rho_{p,\tau}(Y)$, is defined as: 
\begin{equation}\label{eq:Lpquantilescore} 
\rho_{p,\tau}(Y)=\arg\min_{m\in\R}E[\tau((Y-m)_+)^p+(1-\tau)((Y-m)_-)^p], \qquad \textrm{for }p=1,2,\ldots
\end{equation}
where $x_+=\max(x,0)$ and $x_-=\max(-x,0)$, i.e., it is the minimiser of the asymmetric power function $\tau((y-m)_+)^p+(1-\tau)((y-m)_-)^p$,
provided that the expectation exists. It is easy to see that when $p=1$, 
$\rho_{1,\tau}(Y)$ corresponds to the quantile  $q_{\tau}(Y)$ of $Y$ and Equation~\eqref{eq:Lpquantilescore} has a unique solution only if the distribution of $Y$ is strictly increasing in a neighbourhood of $\tau$. When $p=2$, $\rho_{2,\tau}(Y)$ corresponds to the $\tau$ level expectile already introduced by~\cite{NeweyP87}.  Moreover for $p\geq 2$,  using the first order condition,  $L_p$-quantiles can alternatively be written as the  unique solution to the following equation
\begin{equation}\label{eq:Lpquantile}
 \tau E\left[\left((Y-m)_+\right)^{p-1}\right]=(1-\tau)E\left[\left((Y-m)_-\right)^{p-1}\right],\quad\text{for all }\tau\in(0,1).
\end{equation}
From now on we only consider the case $p\geq 2$ unless differently specified.\newline

In the following proposition we summarise some properties of the $L_p$-quantiles that may be useful in the regression framework~\citep[see for instance][]{BernardiBP16} or for risk measurement purposes~\citep[see for example][]{Bellinietal14}. 

\begin{proposition}Let $X,Y$ be random variables with finite moment of order $p-1$. Then, for any  $\tau\in(0,1)$ the following properties are satisfied:
\begin{enumerate}
\item \textit{Monotonicity:} If $X\leq Y$ a.s., then $\rho_{p,\tau}(X)\leq \rho_{p,\tau}(Y)$;
\item \textit{Translation Invariance:} For any $a\in\R$, $\rho_{p,\tau}(X+a)=\rho_{p,\tau}(X)+a$;
\item \textit{Positive homogeneity:} For any $\lambda\geq0$, $\rho_{p,\tau}(\lambda X)=\lambda\rho_{p,\tau}(X)$;
\item \textit{Symmetry:} $\rho_{p,\tau}(X)=-\rho_{p,1-\tau}(-X)$; 
\item \textit{Law-invariance:} If $X$ and $Y$ have the same probability distribution then $\rho_{p,\tau}(X)=\rho_{p,\tau}(Y)$;
\item \textit{Concavity:} For any $\beta\in[0,1]$, $\rho_{p,\tau}(\beta X+(1-\beta)Y)\geq \beta\rho_{p,\tau}(X)+(1-\beta)\rho_{p,\tau}(Y)$ if and only if $p=2$ and $\tau\leq {1}/{2}$.
\end{enumerate}
\end{proposition}
\begin{proof}Properties $1$-$5$ follow easily from the definition of $L_p$-quantiles, hence we only discuss the concavity one. 
From Equation~\eqref{eq:Lpquantile}, $L_p$-quantiles can be written as the unique solution to the equation $E[\varphi(Y-m)]=0$, where the function $\varphi:\R\to\R$ is defined as
$$
 \varphi(t)=\tau\left(t_+\right)^{p-1}-(1-\tau)\left(t_-\right)^{p-1},\quad\text{for  }\tau\in(0,1).
$$
From Proposition 6(b) in~\cite{Bellinietal14}, $L_p$-quantiles are concave if and only if $\varphi$
is concave.  It is easy to see that expectiles, i.e. the $L_p$-quantiles when $p=2$, are concave for $\tau\leq 1/2$. For $p\geq3$, by computing the second derivative of $\varphi$, we note that $\varphi$ is concave for $t<0$ and convex otherwise and never concave on the entire domain for any $\tau$.
\end{proof}
In the next proposition, we provide a further characterisation of the $L_p$-quantiles of a random variable $Y$ with cumulative distribution function $F_Y$ in terms of the raw moments $E[Y^k]$  and  the truncated $k$-th moments $G_{k,Y}$
$$
G_{k,Y}(x)=\int_{-\infty}^x y^k\d F_{Y}(y), \qquad x\in\R,
$$
for $k=0,\ldots,p-1$. {This representation will be exploited throughout the paper to get the main results.}
\begin{proposition}Let $p$  be an odd number, then the $L_p$-quantile can be written as the unique solution  of the following equation {with respect to $m$}:
\begin{equation}\label{eq:polynodd}
\sum_{k=0}^{p-1}\binom{p-1}{k}(-m)^k \left(\tau E[Y^{p-1-k}]-G_{p-1-k,Y}(m)\right)=0.
\end{equation}
Similarly, for $p$ even, the $L_p$-quantile can be written as:
\begin{equation}\label{eq:polyneven}
\sum_{k=0}^{p-1}\binom{p-1}{k}(- m)^{k}\left(\tau E[Y^{p-1-k}]+(1-2\tau)G_{p-1-k,Y}(m)\right)=0.
\end{equation}
\end{proposition}
\begin{proof}
From Equation~\eqref{eq:Lpquantile} and for any $\tau\in(0,1)$:
\begin{equation}\label{eq:focLp}
\tau\int_m^\infty (y-m)^{p-1}\d F_Y(y)=(1-\tau)\int_{-\infty}^m (-1)^{p-1}(y-m)^{p-1}\d F_Y(y).
\end{equation}
Let $p$ be an odd number. Considering that
\begin{equation}\label{eq:foctau}
\tau\int_{-\infty}^\infty (y-m)^{p-1}\d F_Y(y)=\int_{-\infty}^m (y-m)^{p-1}\d F_Y(y)
\end{equation} 
and using the binomial expansion, we get
$$
\tau\left(\int_{-\infty}^\infty\sum_{k=0}^{p-1}\binom{p-1}{k}y^{p-1-k}(-m)^{k}\d F_Y(y)\right)=\left(\int_{-\infty}^m\sum_{k=0}^{p-1}\binom{p-1}{k}y^{p-1-k}(- m)^{k}\d F_Y(y)\right).$$ The linearity of the integral gives 
\begin{equation}\label{eq:polynm}
\sum_{k=0}^{p-1}\binom{p-1}{k}(- m)^{k}\left(\tau E[Y^{p-1-k}]-G_{p-1-k,Y}(m)\right)=0.
\end{equation}
Now let $p$ be an even number; rearranging the terms in Equation~\eqref{eq:focLp} and adding and subtracting $\int_{-\infty}^m (y-m)^{p-1}\d F_Y(y)$, leads to
$$
\tau\left(\int_{-\infty}^\infty(y-m)^{p-1}\d F_Y(y)-2\int_{-\infty}^m (y-m)^{p-1}\d F_Y(y)\right)= \left(\int_{-\infty}^m-(y-m)^{p-1}\d F_Y(y)\right).$$   Using again  the binomial expansion, the equation becomes
$$
\sum_{k=0}^{p-1}\binom{p-1}{k}(- m)^{k}\left(\tau E[Y^{p-1-k}]+(1-2\tau)G_{p-1-k,Y}(m)\right)=0.
$$
\end{proof}

\section{Some useful properties of the Student $t$ distribution}\label{sec:student}
{In this section we gather some results on the Student $t$ distribution that will be used to prove the main theorems of the paper. Let $Y$ be} a Student $t$ random variable with $p$ degrees of freedom having density function
\begin{equation*}
f_Y(y)=C_p\left(1+\frac{y^2}{p}\right)^{-\frac{p+1}{2}},\quad\textrm{where}\quad y\in\R,~C_p=\frac{\Gamma\left(\frac{p+1}{2}\right)}{\sqrt{p\pi}\Gamma\left(\frac{p}{2}\right)}.
\end{equation*} 
When $p$ is even $C_p$ can be written as $C_p={(p-1)!!}p^{-{1}/{2}}/(2(p-2)!!)$ where the symbol $!!$ denotes the double factorial defined as
$$
x!!=\left\{ 
\begin{array}{cc}
x\cdot(x-2)\cdot\ldots \cdot1\qquad&\textrm{for } x \textrm{ positive odd integer},\\
x\cdot(x-2)\cdot\ldots \cdot2\qquad&\textrm{for } x \textrm{ positive even integer},
\end{array}%
\right. 
$$
with the convention that $0!!=0.$
The  raw moments of order $j$, $1\leq j\leq p-1$ are defined as 
$$
 E[Y^j]=\left\{ 
\begin{array}{cc}
0\qquad&\textrm{if } j \textrm{ is odd},\\
\frac{\Gamma\left(\frac{j+1}{2}\right)\Gamma\left(\frac{p-j}{2}\right)}{\sqrt{\pi}\Gamma\left(\frac{p}{2}\right)}p^{\frac{j}{2}}={\prod_{k=1}^{{j}/{2}}\frac{2k-1}{p-2k}p}\qquad&\textrm{if } j \textrm{ is even}.%
\end{array}%
\right. 
$$
In the following proposition we provide an easy way to compute  the truncated moments of order $j$, $G_{j,Y}$. We recall that the symbol $\lfloor\cdot\rfloor$ denotes the lower integer part of a number.
\begin{proposition}Let $Y$ be a Student $t$ random variable with  $p$ degrees of freedom then the truncated moments of order $j$, for $0\leq j\leq p-1$, are given by:
\begin{align*}
G_{j,Y}(m)=-C_p\left(1+\frac{m^2}{p}\right)^{\frac{1-p}{2}}\sum_{i=0}^{\lfloor\frac{j-1}{2}\rfloor}m^{j-1-2i}p^{i+1}\frac{(j-1)!!}{(j-1-2i)!!}{\prod_{k=1}^{i+1}\frac{1}{p-j-2+2k}}+F_Y(m) E[Y^j],
\end{align*}
for $0<j\leq p-1$ and $G_{0,Y}(m)=F_Y(m).$
\end{proposition}
\begin{proof}
It is immediate to see that
\begin{equation*}G_{0,Y}(m)=F_Y(m)\quad \textrm{and}\quad 
G_{1,Y}(m)=\frac{p}{1-p}C_pK_p,
\end{equation*}
where $K_p=\left(1+{m^2}/{p}\right)^{{(1-p)}/{2}}.$
For $j\geq 2$, formulas 2.147(2) and 2.263(1) in~\cite{TableofInt07}, provide (for $p$ odd and even respectively) a characterisation of the truncated moment of order $j$ in terms of the truncated moment of order $j-2$:
\begin{align*}
G_{j,Y}(m)&=-\frac{m^{j-1}}{p-j}pC_pK_p+\frac{p(j-1)}{p-j}G_{j-2,Y}(m).
\end{align*}
By working backwards recursively, we obtain
\begin{equation*}
G_{j,Y}(m)=-C_p\left(1+\frac{m^2}{p}\right)^{\frac{1-p}{2}}\sum_{i=0}^{\lfloor\frac{j-1}{2}\rfloor}m^{j-1-2i}p^{i+1}\frac{(j-1)!!}{(j-1-2i)!!}{\prod_{k=1}^{i+1}\frac{1}{p-j-2+2k}}+F_Y(m) E[Y^j],
\end{equation*}
where we used that $ E[Y^j]=0$ for $j$ odd.
\end{proof}
From this expression it easy to see that 
the truncated moment of order $p-1-k$, evaluated at the $\tau$-quantile, $\tau\in(0,1)$, that is in $m=q_\tau(Y)$, is:
\begin{equation}\label{eq:truncmom}
G_{p-1-k,Y}(m)=-C_pK_p\sum_{i=0}^{\lfloor{\frac{p-k-2}{2}}\rfloor}m^{p-k-2-2i}p^{i+1}\frac{(p-k-2)!!}{(p-k-2-2i)!!}{\prod_{j=1}^{i+1}\frac{1}{k-1+2j}}+\tau E[Y^{p-1-k}]
\end{equation}
for $0\leq k<p-1$, and  $G_{0,Y}(m)=\tau$.


Another interesting property of the Student $t$ distribution, which will be extensively used to prove our results, is related to its cumulative distribution function. In what follows we report an interesting formula  provided by~\cite{Shaw05} which states that  the cumulative distribution function $F_Y$ for a Student $t$ with even degrees of freedom $p$ can be written as:
\begin{equation}
F_Y(m)=mK_p\left(\sum_{i=0}^{\frac{p}{2}-1}{m^{2i}}a_p(i)\right)+\frac{1}{2},
\end{equation} 
where
$$a_p(i)=C_p\prod_{j=0}^i\frac{p-2j}{(2j+1)p}.$$
By denoting $A_{p-1}$  the double factorial binomial coefficient 
\begin{equation*}
A_{p-1}(k)=\frac{(p-1)!!}{k!!(p-1-k)!!}
\end{equation*}
having the property $A_{p-1}(k)=A_{p-1}(p-1-k)$,  the equality $F_Y(m)=\tau$ can be written as
\begin{equation}\label{eq:cdfStudent0}
mK_p\left(\sum_{i=0}^{\frac{p}{2}-1}{m^{2i}}\frac{1}{2}A_{p-1}(2i+1)p^{-\frac{2i+1}{2}}\right)=\tau-\frac{1}{2}.
\end{equation}
 
Equation~\eqref{eq:cdfStudent0} can be rearranged to:
\begin{equation}\label{eq:cdfStudent1}
\frac{2}{K_p}\left(\tau-\frac{1}{2}\right)=\sum_{i=0}^{\frac{p}{2}-1}{m^{2i+1}}p^{-\frac{2i+1}{2}}A_{p-1}(2i+1).
\end{equation}
By squaring it 
and by changing the index of the summation with $k={p}/{2}-1-i$ it  can be  rewritten as
\begin{equation}\label{eq:cdfStudentsquared}
\tau(1-\tau)=\frac{1}{4}-\frac{K_p^2p^{1-p}}{4}\left(\sum_{k=0}^{\frac{p}{2}-1}{m^{p-1-2k}}p^{k}A_{p-1}(2k)\right)^2.
\end{equation}

\section{$L_p$-quantiles for the Student $t$ distribution}\label{sec:Lpstudent}
In his paper~\cite{Koenker93} investigates the relation between quantiles and expectiles. In details he shows that for a particular class of distributions (that are affine transformations of a Student $t$ random variable with two degrees of freedom) the $\tau$ expectiles and the $\tau$ quantiles coincide for every $\tau$. After that~\cite{Remillard95} give sufficient conditions under which a quantile and an expectile coincide. Those results were then extended by~\cite{Zou14} that analytically characterised distributions whose $\omega(\beta)$ expectile and $\beta$ quantile coincide, for any monotone function $\omega(\cdot)$.
In what follows we extend all of those results and prove that  the Student $t$ distribution with  $p$ degrees of freedom has the peculiarity of having the $L_p$-quantile and the quantile coinciding for every $\tau\in(0,1)$. This result is then extended to any affine transformation of the Student $t$ random variable.

We first consider the case where $p$ is an odd number.
\begin{theorem}Let $p$ be an odd number and $Y$ a random variable with Student $t$ distribution with $p$ degrees of freedom. Then the  $L_p$-quantile $\rho_{p,\tau}(Y)$ and the quantile $q_\tau(Y)$ coincide for every $\tau\in(0,1)$.
\end{theorem}
\begin{proof}  As already said, when $p=1$ we have $\rho_{1,\tau}(Y)=q_\tau(Y)$, hence we consider the case $p\geq 3$. If the $L_p$-quantile and the quantile coincide then the equation
\begin{equation}\label{eq:poddbiss}
\sum_{k=0}^{p-1}B_{p-1}(k)(-m)^k \left(\tau E[Y^{p-1-k}]-G_{p-1-k,Y}(m)\right)=0
\end{equation}
is  satisfied for $m=q_{\tau}(Y)=\rho_{p,\tau}(Y)$ where $B_{p-1}(k)=\binom{p-1}{k}=\frac{(p-1)!}{k!(p-1-k)!}$.  Hence from now on we assume $m=q_{\tau}(Y).$ 
Using Equation~\eqref{eq:truncmom} to write $G_{p-1-k,Y}(m)$ in~\eqref{eq:poddbiss} we can cancel the term $\tau E[Y^{p-1-k}]$ 
so that the equation to verify becomes
\begin{equation}
\sum_{k=0}^{p-2}B_{p-1}(k)(-m)^k\left(C_pK_p\sum_{i=0}^{\lfloor{\frac{p-k-2}{2}}\rfloor}m^{p-k-2-2i}p^{i+1}\frac{(p-k-2)!!}{(p-k-2-2i)!!}{\prod_{j=1}^{i+1}\frac{1}{k-1+2j}}\right)=0,
\end{equation}
which, using the properties of the binomial coefficients, can be rewritten as
\begin{equation}
K_p\frac{p^{\frac{1}{2}}}{\pi}\sum_{k=0}^{p-2} \sum_{i=0}^{\lfloor{\frac{p-k-2}{2}}\rfloor}(-1)^km^{p-2-2i}p^{i}A_{p-1}(k)A_{p-1}(k+1+2i)=0.
\end{equation}
By dividing both sides for $K_p{p^{{1}/{2}}}/{\pi}$ and inverting the order of the summations, we obtain  
\begin{equation}\label{eq:poddbis}
\sum_{i=0}^{{\frac{p-3}{2}}}m^{p-2-2i}p^i\sum_{k=0}^{p-2(i+1)}(-1)^kA_{p-1}(k)A_{p-1}(k+1+2i)=0.
\end{equation}
Equation~\eqref{eq:poddbis} is a polynomial equation of order $p-2$ containing only odd powers of $m$. The proof is concluded by showing that  the coefficients of these powers are all null. Indeed the function
$$
f_i(k):=(-1)^kA_{p-1}(k)A_{p-1}(k+1+2i),\quad\text{for }i=0,\ldots,\frac{p-3}{2}
$$
satisfies $f_i(k)=-f_i(p-2(i+1)-k)$ for $k=0,\ldots,p-2(i+1)$, hence all the terms in the second summation in Equation~\eqref{eq:poddbis}
cancel pairwise, the first with the last one, the second with the second to last and so on.
\end{proof}
The following theorem proves the analogous result for $p$ even.
\begin{theorem}Let $p$ be an even number and $Y$ a random variable with Student $t$ distribution with $p$ degrees of freedom. Then the  $L_p$-quantile $\rho_{p,\tau}(Y)$ and the quantile $q_\tau(Y)$ coincide for every $\tau\in(0,1)$.
\end{theorem}
\begin{proof}
If the $L_p$-quantile and the quantile coincide then the equation
\begin{equation}\label{eq:peven}
\sum_{k=0}^{p-1}B_{p-1}(k)(-m)^k\left(\tau E[Y^{p-1-k}]+(1-2\tau)G_{p-1-k,Y}(m)\right)=0,
\end{equation}
is satisfied by $m=q_\tau(Y)=\rho_{p,\tau}(Y),$ for any $\tau\in(0,1)$. For $\tau=1/2$ the result is immediate, hence we only consider  $\tau\in(0,1)\backslash\{1/2\}$.


Using Equation~\eqref{eq:truncmom}, the equation to solve becomes
\begin{align}\label{eq:even}
&(2\tau-1)\sum_{k=0}^{p-2}B_{p-1}(k)(-m)^k\left(C_pK_p\sum_{i=0}^{\lfloor{\frac{p-k-2}{2}}\rfloor}m^{p-k-2-2i}p^{i+1}\frac{(p-k-2)!!}{(p-k-2-2i)!!}{\prod_{j=1}^{i+1}\frac{1}{k-1+2j}}\right)\notag\\
&+2\tau(1-\tau)\sum_{k=0}^{p-1}B_{p-1}(k)(-m)^k E[Y^{p-1-k}]=0.
\end{align}
For $k=0,\ldots,p-2$, 
$$
 E[Y^{p-1-k}]=\left\{ 
\begin{array}{cc}
0\qquad&\textrm{if } k \textrm{ is even},\\
\prod_{j=1}^{\frac{p-1-k}{2}}\frac{2j-1}{p-2j}p\qquad&\textrm{if } k \textrm{ is odd},%
\end{array}%
\right. 
$$
thus using the properties of the factorial, Equation~\eqref{eq:even} can be rewritten  as
\begin{align*}
&\left(\tau-\frac{1}{2}\right)\frac{K_p}{2}\sum_{k=0}^{p-2} (-1)^k\left(\sum_{i=0}^{\lfloor{\frac{p-k-2}{2}}\rfloor}m^{p-2-2i}p^{i+\frac{1}{2}}A_{p-1}(k)A_{p-1}(k+1+2i)\right)\notag\\
&=\tau(1-\tau)p^{\frac{p-1}{2}}\sum_{k=1,~k~odd}^{p-1}m^k p^{-\frac{k}{2}} A_{p-1}(k).
\end{align*}
By applying the equalities~\eqref{eq:cdfStudent1} and~\eqref{eq:cdfStudentsquared} 
the equation to solve becomes
\begin{align*}
&{K_p^2}\sum_{k=0}^{p-2} (-1)^k\left(\sum_{i=0}^{\lfloor{\frac{p-k-2}{2}}\rfloor}m^{p-2-2i}p^{-\frac{p-2-2i}{2}}A_{p-1}(k)A_{p-1}(k+1+2i)\right)\notag\\
&=1-K_p^2p^{1-p}\left(\sum_{k=0}^{\frac{p}{2}-1}{m^{p-1-2k}}p^{k}A_{p-1}(2k)\right)^2.
\end{align*}
Inverting the order of the two summations in the left hand side of the above equation, 
 and rearranging the terms, we obtain
\begin{align}\label{eq:withsquaredterm}
&\sum_{i=0}^{\frac{p}{2}-1}m^{-2i}p^{i}\sum_{k=0}^{p-2(i+1)} (-1)^kA_{p-1}(k)A_{p-1}(k+1+2i)+\left(\sum_{k=0}^{\frac{p}{2}-1}{m^{\frac{p}{2}-2k}}p^{-\frac{p}{4}+k}A_{p-1}(2k)\right)^2\\
&=K_p^{-2}m^{2-p}p^{\frac{p}{2}-1}.\notag
\end{align}
For $p=2$, the equation in~\eqref{eq:withsquaredterm} is easily verified, hence we only consider the case $p\geq 4$.
Using the fact that 
\begin{align}\label{squaredterm}
&\left(\sum_{k=0}^{\frac{p}{2}-1}{m^{\frac{p}{2}-2k}}p^{-\frac{p}{4}+k}A_{p-1}(2k)\right)^2=\\
&\sum_{j=0}^{\frac{p}{2}-2}m^{-2j}p^{j}\sum_{\substack{h=1,\\h~\text{odd}}}^{p-2j-3}A(h)A(h+2j+1)+\sum_{j=1}^{\frac{p}{2}}m^{2j}p^{-j}\sum_{\substack{h=0,\\h~\text{even}}}^{{p}-2j}A(h)A(h+2j-1)
\end{align}
(for the proof of this result see the Appendix),  
it follows that Equation~\eqref{eq:withsquaredterm}  becomes
\begin{align}\label{eq:withsquaredterm2}
&\sum_{i=0}^{\frac{p}{2}-1}m^{-2i}p^{i}\sum_{\substack{k=0,\\
k~even}}^{p-2(i+1)} A_{p-1}(k)A_{p-1}(k+1+2i)+\sum_{j=1}^{\frac{p}{2}}m^{2j}p^{-j}\sum_{\substack{h=0,\\h~\text{even}}}^{{p}-2j}A(h)A(h+2j-1)\\
&=K_p^{-2}m^{2-p}p^{\frac{p}{2}-1}.\notag
\end{align}
Using formula (27) in~\cite{GouldQ12}, for which given an even number $l$ we have $A_{p-1}(l)=B_{\frac{p-1}{2}}({l}/{2})$ 
the summations for $k$ and $h$ even can be rewritten respectively as
$$
\sum_{\substack{k=0,\\k~even}}^{p-2(i+1)} A_{p-1}(k)A_{p-1}(k+1+2i)=\sum_{j=0}^{\frac{p}{2}-(i+1)} B_{\frac{p-1}{2}}(j)B_\frac{p-1}{2}\left(\frac{p}{2}-(i+1)-j\right)=B_{p-1}\left(\frac{p}{2}-1-i\right),
$$
and
$$
\sum_{\substack{h=0,\\h~even}}^{p-2j} A_{p-1}(h)A_{p-1}(h+2j-1)=\sum_{h=0}^{\frac{p}{2}-j} B_{\frac{p-1}{2}}(h)B_\frac{p-1}{2}\left(\frac{p}{2}-j-h\right)=B_{p-1}\left(\frac{p}{2}-j\right),
$$
where the last equality is obtained  using the  Chu-Vandermonde binomial identity, see for instance~\cite{Askey75}, pages 59-60. It follows that Equation~\eqref{eq:withsquaredterm} becomes  
\begin{align}\label{eq:withsquaredterm3}
&\sum_{i=0}^{\frac{p}{2}-1}m^{-2i}p^{i}B_{p-1}\left(\frac{p}{2}-1-i\right)+\sum_{j=1}^{\frac{p}{2}}m^{2j}p^{-j}B_{p-1}\left(\frac{p}{2}-1+j\right)\\
&=K_p^{-2}m^{2-p}p^{\frac{p}{2}-1}.\notag
\end{align}
The proof is concluded by  decomposing the term $K_p^{-2}=(1+{m^2}/{p})^{p-1}$ using the binomial expansion
$$
K_p^{-2}p^{\frac{p}{2}-1}m^{2-p}=\sum_{i=0}^{p-1}B_{p-1}(i)m^{-p+2+2i}p^{\frac{p}{2}-1-i}=\sum_{i=-\frac{p}{2}}^{\frac{p}{2}-1}m^{-2i}p^{i}B_{p-1}\left(\frac{p}{2}-1-i\right).\\
$$
\end{proof}
\begin{corollary}Let $X$ be any affine transformation of a random variable with Student $t$ distribution with $p$ degrees of freedom, that is $X\stackrel{d}{=}a+bY$, $a\in\R,~b>0$, then $q_{\tau}(X)=\rho_{p,\tau}(X)$.
\end{corollary}
\begin{proof}The proof follows immediately from the translation invariance and positive homogeneity of the $L_p$-quantiles.
\end{proof}
\appendix
\section*{Appendix}\label{withsquaredterm}
Let us take a closer look at the squared term in Equation~\eqref{eq:withsquaredterm}
\begin{align}
&\left(\sum_{k=0}^{\frac{p}{2}-1}{m^{\frac{p}{2}-2k}}p^{-\frac{p}{4}+k}A_{p-1}(2k)\right)^2=\notag\\
&\left(\sum_{k=0}^{\frac{p}{2}-1}{m^{{p}-4k}}p^{-\frac{p}{2}+2k}A_{p-1}(2k)^2\right)+2\sum_{k=0}^{\frac{p}{2}-2}\sum_{j=k+1}^{\frac{p}{2}-1}{m^{\frac{p}{2}-2(k+j)}}p^{-\frac{p}{2}+k+j}A_{p-1}(2k)A_{p-1}(2j)=\label{changeindex}\\
&m^pp^{-\frac{p}{2}}+m^{-p+4}p^{\frac{p}{2}-2}A(1)^2+\sum_{i=1}^{p-3}m^{p-2i}p^{-\frac{p}{2}+i}\left(A(i)^2\mathbb{I}_{\{i~\text{even}\}}+\sum_{\substack{h=\max\{0,2i+2-p\},\\h~\text{even}}}^{i-1}2A(h)A(2i-h)\right)=\notag\\
&m^pp^{-\frac{p}{2}}+m^{-p+4}p^{\frac{p}{2}-2}A(1)^2+\sum_{j=1-\frac{p}{2}}^{\frac{p}{2}-3}m^{-2j}p^{j}\left(A(j+\frac{p}{2})^2\mathbb{I}_{\{j+\frac{p}{2}~\text{even}\}}+\sum_{\substack{h=\max\{0,2j+2\},\\h~\text{even}}}^{\frac{p}{2}+j-1}2A(h)A(2j+p-h)\right)=\notag\\
&m^{-p+4}p^{\frac{p}{2}-2}A(1)^2+\sum_{j=0}^{\frac{p}{2}-3}m^{-2j}p^{j}\left(A(j+\frac{p}{2})^2\mathbb{I}_{\{j+\frac{p}{2}~\text{even}\}}+\sum_{\substack{k=0,\\k~\text{even}}}^{\frac{p}{2}-j-3}2A(k+2j+2)A(p-k-2)\right)\\
&+m^pp^{-\frac{p}{2}}+\sum_{j=1}^{\frac{p}{2}-1}m^{2j}p^{-j}\left(A(\frac{p}{2}-j)^2\mathbb{I}_{\{\frac{p}{2}-j~\text{even}\}}+\sum_{\substack{h=0,\\h~\text{even}}}^{\frac{p}{2}-j-1}2A(h)A(-2j+p-h)\right)=\notag\\
&m^{-p+4}p^{\frac{p}{2}-2}A(1)^2+\sum_{j=0}^{\frac{p}{2}-3}m^{-2j}p^{j}\left(A(j+\frac{p}{2})^2\mathbb{I}_{\{j+\frac{p}{2}~\text{even}\}}+\sum_{\substack{h=1,\\h~\text{odd}}}^{\frac{p}{2}-j-2}2A(h)A(h+2j+1)\right)\\
&+m^pp^{-\frac{p}{2}}+\sum_{j=1}^{\frac{p}{2}-1}m^{2j}p^{-j}\left(A(\frac{p}{2}-j)^2\mathbb{I}_{\{\frac{p}{2}-j~\text{even}\}}+\sum_{\substack{h=0,\\h~\text{even}}}^{\frac{p}{2}-j-1}2A(h)A(h+2j-1)\right)=\notag\\
&\sum_{j=0}^{\frac{p}{2}-2}m^{-2j}p^{j}\sum_{\substack{h=1,\\h~\text{odd}}}^{p-2j-3}A(h)A(h+2j+1)+\sum_{j=1}^{\frac{p}{2}}m^{2j}p^{-j}\sum_{\substack{h=0,\\h~\text{even}}}^{{p}-2j}A(h)A(h+2j-1)\notag.
\end{align}
\bibliographystyle{chicago}

\begin{thebibliography}{}

\bibitem[\protect\citeauthoryear{Abdous and R{{\'e}}millard}{Abdous and
  R{{\'e}}millard}{1995}]{Remillard95}
Abdous, B. and B.~R{{\'e}}millard (1995).
\newblock Relating quantiles and expectiles under weighted-symmetry.
\newblock {\em Ann. Inst. Statist. Math.\/}~{\em 47\/}(2), 371--384.

\bibitem[\protect\citeauthoryear{Artzner, Delbaen, Eber, and Heath}{Artzner
  et~al.}{1999}]{Artzner99}
Artzner, P., F.~Delbaen, J.-M. Eber, and D.~Heath (1999).
\newblock Coherent measures of risk.
\newblock {\em Math. Finance\/}~{\em 9\/}(3), 203--228.

\bibitem[\protect\citeauthoryear{Askey}{Askey}{1975}]{Askey75}
Askey, R. (1975).
\newblock {\em Orthogonal polynomials and special functions}.
\newblock Society for Industrial and Applied Mathematics, Philadelphia, Pa.

\bibitem[\protect\citeauthoryear{{Basel Committee on Banking
  Supervision}}{{Basel Committee on Banking Supervision}}{2013}]{BCBS13}
{Basel Committee on Banking Supervision} (2013).
\newblock Consultative {D}ocument {O}ctober 2013. {F}undamental review of the
  trading book: {A} revised market risk framework.

\bibitem[\protect\citeauthoryear{Bellini and Bignozzi}{Bellini and
  Bignozzi}{2015}]{BelliniB15}
Bellini, F. and V.~Bignozzi (2015).
\newblock On elicitable risk measures.
\newblock {\em Quant. Finance\/}~{\em 15\/}(5), 725--733.

\bibitem[\protect\citeauthoryear{Bellini and Di~Bernardino}{Bellini and
  Di~Bernardino}{2015}]{BelliniDB15}
Bellini, F. and E.~Di~Bernardino (2015).
\newblock Risk management with expectiles.
\newblock To appear in \emph{The European Journal of Finance}.

\bibitem[\protect\citeauthoryear{Bellini, Klar, M{{\"u}}ller, and
  Rosazza~Gianin}{Bellini et~al.}{2014}]{Bellinietal14}
Bellini, F., B.~Klar, A.~M{{\"u}}ller, and E.~Rosazza~Gianin (2014).
\newblock Generalized quantiles as risk measures.
\newblock {\em Insurance Math. Econom.\/}~{\em 54}, 41--48.

\bibitem[\protect\citeauthoryear{Bernardi, Bignozzi, and Petrella}{Bernardi
  et~al.}{2016}]{BernardiBP16}
Bernardi, M., V.~Bignozzi, and L.~Petrella (2016).
\newblock Bayesian inference for ${L}_p$ quantile regression models.
\newblock Working papers, Sapienza University of Rome.

\bibitem[\protect\citeauthoryear{Bernardi, Bottone, and Petrella}{Bernardi
  et~al.}{2016}]{Bernardietal16}
Bernardi, M., M.~Bottone, and L.~Petrella (2016).
\newblock Bayesian robust quantile regression.
\newblock arXiv:1605.05602.

\bibitem[\protect\citeauthoryear{Bernardi, Gayraud, and Petrella}{Bernardi
  et~al.}{2015}]{Bernardietal15}
Bernardi, M., G.~Gayraud, and L.~Petrella (2015).
\newblock Bayesian tail risk interdependence using quantile regression.
\newblock {\em Bayesian Anal.\/}~{\em 10\/}(3), 553--603.

\bibitem[\protect\citeauthoryear{Breckling and Chambers}{Breckling and
  Chambers}{1988}]{BrecklingChambers88}
Breckling, J. and R.~Chambers (1988).
\newblock {$M$}-quantiles.
\newblock {\em Biometrika\/}~{\em 75\/}(4), 761--771.

\bibitem[\protect\citeauthoryear{Chambers and Tzavidis}{Chambers and
  Tzavidis}{2006}]{ChambersT06}
Chambers, R. and N.~Tzavidis (2006).
\newblock {$M$}-quantile models for small area estimation.
\newblock {\em Biometrika\/}~{\em 93\/}(2), 255--268.

\bibitem[\protect\citeauthoryear{Chen}{Chen}{1996}]{Chen96}
Chen, Z. (1996).
\newblock Conditional ${L}_p$-quantiles and their application to the testing of
  symmetry in non-parametric regression.
\newblock {\em Statist. Probab. Lett.\/}~{\em 29\/}(2), 107--115.

\bibitem[\protect\citeauthoryear{De~Rossi and Harvey}{De~Rossi and
  Harvey}{2009}]{Derossi09}
De~Rossi, G. and A.~Harvey (2009).
\newblock Quantiles, expectiles and splines.
\newblock {\em J. Econometrics\/}~{\em 152\/}(2), 179--185.

\bibitem[\protect\citeauthoryear{Deprez and Gerber}{Deprez and
  Gerber}{1985}]{DeprezG85}
Deprez, O. and H.~U. Gerber (1985).
\newblock On convex principles of premium calculation.
\newblock {\em Insurance Math. Econom.\/}~{\em 4\/}(3), 179--189.

\bibitem[\protect\citeauthoryear{Efron}{Efron}{1991}]{Efron91}
Efron, B. (1991).
\newblock Regression percentiles using asymmetric squared error loss.
\newblock {\em Statist. Sinica\/}~{\em 1\/}(1), 93--125.

\bibitem[\protect\citeauthoryear{Embrechts, Puccetti, R{\"u}schendorf, Wang,
  and Beleraj}{Embrechts et~al.}{2014}]{Embrechts14}
Embrechts, P., G.~Puccetti, L.~R{\"u}schendorf, R.~Wang, and A.~Beleraj (2014).
\newblock An academic response to {B}asel 3.5.
\newblock {\em Risks\/}~{\em 2\/}(1), 25--48.

\bibitem[\protect\citeauthoryear{Emmer, Kratz, and Tasche}{Emmer
  et~al.}{2015}]{Emmer16}
Emmer, S., M.~Kratz, and D.~Tasche (2015).
\newblock What is the best risk measure in practice? a comparison of standard
  measures.
\newblock {\em Journal of Risk\/}~{\em 18\/}(2), 31--60.

\bibitem[\protect\citeauthoryear{F\"{o}llmer and Schied}{F\"{o}llmer and
  Schied}{2002}]{FollmerS2002}
F\"{o}llmer, H. and A.~Schied (2002).
\newblock Convex measures of risk and trading constraints.
\newblock {\em Finance Stoch.\/}~{\em 6\/}(4), 429--47.

\bibitem[\protect\citeauthoryear{Gneiting}{Gneiting}{2011}]{Gneiting11}
Gneiting, T. (2011).
\newblock Making and evaluating point forecast.
\newblock {\em J. Amer. Statis. Assoc.\/}~{\em 106\/}(494), 746--762.

\bibitem[\protect\citeauthoryear{Gould and Quaintance}{Gould and
  Quaintance}{2012}]{GouldQ12}
Gould, H. and J.~Quaintance (2012).
\newblock Double fun with double factorials.
\newblock {\em Math. Mag.\/}~{\em 85\/}(3), 177--192.

\bibitem[\protect\citeauthoryear{Gradshteyn and Ryzhik}{Gradshteyn and
  Ryzhik}{2007}]{TableofInt07}
Gradshteyn, I.~S. and I.~M. Ryzhik (2007).
\newblock {\em Table of integrals, series, and products\/} (Seventh ed.).
\newblock Elsevier/Academic Press, Amsterdam.
\newblock Translated from the Russian, Translation edited and with a preface by
  Alan Jeffrey and Daniel Zwillinger.

\bibitem[\protect\citeauthoryear{Guo, Zhou, Huang, and H{\"a}rdle}{Guo
  et~al.}{2015}]{Guoetal13}
Guo, M., L.~Zhou, J.~Z. Huang, and W.~K. H{\"a}rdle (2015).
\newblock Functional data analysis of generalized regression quantiles.
\newblock {\em Statist. Comput.\/}~{\em 25\/}(2), 189--202.

\bibitem[\protect\citeauthoryear{Heinrich}{Heinrich}{2013}]{Heinrich13}
Heinrich, C. (2013).
\newblock The mode functional is not elicitable.
\newblock {\em Biometrika\/}~{\em 101\/}(1), 245--251.

\bibitem[\protect\citeauthoryear{IAIS}{IAIS}{2014}]{IAIS14}
IAIS (2014).
\newblock Consultation document december 2014. risk-based global insurance
  capital standard.

\bibitem[\protect\citeauthoryear{Jakobsons and Vanduffel}{Jakobsons and
  Vanduffel}{2015}]{JakobsonsV15}
Jakobsons, E. and S.~Vanduffel (2015).
\newblock Dependence uncertainty bounds for the expectile of a portfolio.
\newblock {\em Risks\/}~{\em 3\/}(4), 599--623.

\bibitem[\protect\citeauthoryear{Jones}{Jones}{1994}]{Jones94}
Jones, M.~C. (1994).
\newblock Expectiles and m-quantiles are quantiles.
\newblock {\em Statist. Probab. Lett.\/}~{\em 20\/}(2), 149--153.

\bibitem[\protect\citeauthoryear{Jorion}{Jorion}{2007}]{Jorion07}
Jorion, P. (2007).
\newblock {\em Value at risk: the new benchmark for managing financial risk},
  Volume~3.
\newblock McGraw-Hill New York.

\bibitem[\protect\citeauthoryear{Koenker}{Koenker}{1992}]{Koenker92}
Koenker, R. (1992).
\newblock When are expectiles percentiles?
\newblock {\em Econometric Theory\/}~{\em 8}, 423--424.

\bibitem[\protect\citeauthoryear{Koenker}{Koenker}{1993}]{Koenker93}
Koenker, R. (1993).
\newblock When are expectiles percentiles?
\newblock {\em Econometric Theory\/}~{\em 9\/}(3), 526--527.

\bibitem[\protect\citeauthoryear{Koenker}{Koenker}{2005}]{Koenker05}
Koenker, R. (2005).
\newblock {\em Quantile regression}, Volume~38 of {\em Econometric Society
  Monographs}.
\newblock Cambridge University Press, Cambridge.

\bibitem[\protect\citeauthoryear{Koenker and Bassett}{Koenker and
  Bassett}{1978}]{KoenkerBassett78}
Koenker, R. and G.~Bassett, Jr. (1978).
\newblock Regression quantiles.
\newblock {\em Econometrica\/}~{\em 46\/}(1), 33--50.

\bibitem[\protect\citeauthoryear{Kottas and Krnjaji{{\'c}}}{Kottas and
  Krnjaji{{\'c}}}{2009}]{KottasK09}
Kottas, A. and M.~Krnjaji{{\'c}} (2009).
\newblock Bayesian semiparametric modelling in quantile regression.
\newblock {\em Scand. J. Stat.\/}~{\em 36\/}(2), 297--319.

\bibitem[\protect\citeauthoryear{Lambert, Pennock, and Shoham}{Lambert
  et~al.}{2008}]{Lambertetal08}
Lambert, N.~S., D.~M. Pennock, and Y.~Shoham (2008).
\newblock Eliciting properties of probability distributions.
\newblock In {\em Proceedings of the 9th ACM Conference on Electronic
  Commerce}, pp.\  129--138. ACM.

\bibitem[\protect\citeauthoryear{Newey and Powell}{Newey and
  Powell}{1987}]{NeweyP87}
Newey, W.~K. and J.~L. Powell (1987).
\newblock Asymmetric least squares estimation and testing.
\newblock {\em Econometrica\/}~{\em 55\/}(4), 819--847.

\bibitem[\protect\citeauthoryear{Pratesi, Ranalli, and Salvati}{Pratesi
  et~al.}{2009}]{Pratesi09}
Pratesi, M., M.~G. Ranalli, and N.~Salvati (2009).
\newblock Nonparametric m-quantile regression using penalised splines.
\newblock {\em J. Nonparametr. Stat.\/}~{\em 21\/}(3), 287--304.

\bibitem[\protect\citeauthoryear{Savage}{Savage}{1971}]{Savage71}
Savage, L.~J. (1971).
\newblock Elicitation of personal probabilities and expectations.
\newblock {\em J. Amer. Statist. Assoc.\/}~{\em 66}, 783--801.

\bibitem[\protect\citeauthoryear{Schnabel and Eilers}{Schnabel and
  Eilers}{2009}]{Schnabel09}
Schnabel, S.~K. and P.~H. Eilers (2009).
\newblock Optimal expectile smoothing.
\newblock {\em Comput. Statist. Data Anal.\/}~{\em 53\/}(12), 4168--4177.

\bibitem[\protect\citeauthoryear{Shaw}{Shaw}{2005}]{Shaw05}
Shaw, W.~T. (2005).
\newblock New methods for simulating the student t-distribution-direct use of
  the inverse cumulative distribution.
\newblock {\em J. Comput. Finance\/}~{\em 9}, 37--73.

\bibitem[\protect\citeauthoryear{Sobotka, Kauermann, Waltrup, and
  Kneib}{Sobotka et~al.}{2013}]{Sobotka13}
Sobotka, F., G.~Kauermann, L.~S. Waltrup, and T.~Kneib (2013).
\newblock On confidence intervals for semiparametric expectile regression.
\newblock {\em Stat. Comput.\/}~{\em 23\/}(2), 135--148.

\bibitem[\protect\citeauthoryear{Sobotka and Kneib}{Sobotka and
  Kneib}{2012}]{Sobotka12}
Sobotka, F. and T.~Kneib (2012).
\newblock Geoadditive expectile regression.
\newblock {\em Comput. Statist. Data Anal.\/}~{\em 56\/}(4), 755--767.

\bibitem[\protect\citeauthoryear{Taddy and Kottas}{Taddy and
  Kottas}{2010}]{TaddyK10}
Taddy, M.~A. and A.~Kottas (2010).
\newblock A {B}ayesian nonparametric approach to inference for quantile
  regression.
\newblock {\em J. Bus. Econom. Statist.\/}~{\em 28\/}(3), 357--369.

\bibitem[\protect\citeauthoryear{Taylor}{Taylor}{2008}]{Taylor08}
Taylor, J.~W. (2008).
\newblock Estimating value at risk and expected shortfall using expectiles.
\newblock {\em Journal of Financial Econometrics\/}~{\em 6\/}(2), 231--252.

\bibitem[\protect\citeauthoryear{Yu and Moyeed}{Yu and Moyeed}{2001}]{YuM01}
Yu, K. and R.~A. Moyeed (2001).
\newblock Bayesian quantile regression.
\newblock {\em Statist. Probab. Lett.\/}~{\em 54\/}(4), 437--447.

\bibitem[\protect\citeauthoryear{Ziegel}{Ziegel}{2014}]{Ziegel14}
Ziegel, J. (2014).
\newblock Coherency and elicitability.
\newblock To appear in \emph{Math. Finance}.

\bibitem[\protect\citeauthoryear{Zou}{Zou}{2014}]{Zou14}
Zou, H. (2014).
\newblock Generalizing {K}oenker's distribution.
\newblock {\em J. Statist. Plann. Inference\/}~{\em 148}, 123--127.

\end{thebibliography}

 \end{document}